\documentclass[journal,12pt,draftcls,onecolumn,a4paper]{IEEEtran}
\IEEEoverridecommandlockouts
\usepackage{graphicx}
\usepackage{amsmath,bbm,epsfig,amssymb,amsfonts, amstext, verbatim,amsopn,cite,subfigure,multirow,multicol,lipsum}
\usepackage{balance}
\usepackage{booktabs,threeparttable}
\usepackage{amsfonts}
\usepackage{enumerate}   
\newcommand{\R}{\mathbb{R}}
\usepackage{epsfig}
\usepackage{epstopdf}
\usepackage{setspace}
\usepackage{stmaryrd}
\usepackage{psfrag}	
\usepackage{multirow}
\usepackage{float}
\usepackage[process=auto]{pstool}
\usepackage{etoolbox}
\usepackage{algorithm}
\usepackage{algorithmic}
\usepackage{hyperref}
\usepackage{pifont}
\allowdisplaybreaks
\usepackage{amsmath,amssymb,amsfonts}
\usepackage{algorithmic}
\usepackage{graphicx}
\usepackage{textcomp}
\usepackage{bm}
\usepackage{amsmath,bbm,amssymb,amsfonts,amstext,amsopn}
\usepackage{amssymb}

\usepackage{amsthm}
\DeclareMathOperator{\Tr}{Tr}
\DeclareMathOperator{\Rank}{Rank}
\usepackage{cite}
\usepackage{balance}
\allowbreak
\usepackage{url}
\usepackage{epsfig}
\usepackage{setspace}
\usepackage{stmaryrd}
\usepackage{psfrag}	
\usepackage{float}
\usepackage{amsthm}

\newcommand{\subparagraph}{}
\usepackage[compact]{titlesec}
\usepackage[english]{babel}
\usepackage[process=auto]{pstool}
\usepackage{amsthm}

\theoremstyle{remark}
\newtheorem{remark}{Remark}
\usepackage{epstopdf}
\usepackage{algorithmic}
\usepackage{xcolor}
\usepackage{etoolbox}
\usepackage{algorithm}
\usepackage[process=auto]{pstool}
\usepackage{epsfig}
\newtheorem{thm}{Theorem}

\newtheorem{prop}{Proposition}

\begin{document}
\title{Resource Allocation for Multi-User Downlink MISO OFDMA-URLLC Systems}
\author{\IEEEauthorblockN{Walid R. Ghanem,Vahid Jamali, Yan Sun, and Robert Schober}
		\thanks{This paper was presented in part at IEEE ICC 2019 \cite{ghanem1}.}
	\thanks{The authors are with the Institute for Digital Communications, Friedrich-Alexander-University Erlangen-N\"urnberg (FAU), Germany (email: \{walid.ghanem, vahid.jamali, yan.sun, and robert.schober\}@fau.de).}}

\maketitle
\begin{abstract}
	This paper considers the resource allocation algorithm design for downlink multiple-input single-output (MISO) orthogonal frequency division multiple access (OFDMA) ultra-reliable low latency communication (URLLC) systems. To meet the stringent delay requirements of URLLC, short packet transmission is adopted and taken into account for resource allocation algorithm design. The resource allocation is optimized for maximization of the weighted system sum throughput subject to quality-of-service (QoS) constraints regarding the URLLC users' number of transmitted bits, packet error probability, and  delay. Despite the non-convexity of the resulting optimization problem, the optimal solution is found via monotonic optimization. The corresponding optimal resource allocation policy can serve as a performance upper bound for sub-optimal low-complexity solutions. We develop such a low-complexity resource allocation algorithm to strike a balance between performance and complexity. Our simulation results reveal the importance of using multiple antennas for reducing the latency and improving the reliability of URLLC systems. Moreover, the proposed sub-optimal algorithm is shown to closely approach the performance of the proposed optimal algorithm and outperforms two baseline schemes by a considerable margin, especially when the users have heterogeneous delay requirements. Finally, conventional resource allocation designs based on Shannon's capacity formula are shown to be not applicable in MISO OFDMA-URLLC systems as they  may violate the users'  delay constraints. 
\end{abstract}
 \section{Introduction} 
The fifth-generation (5G) wireless communication networks impose several different system design objectives including high data rates, high spectral efficiency, reduced latency, higher system capacity, and massive device connectivity. One important objective is to enable ultra-reliable low latency communication (URLLC). URLLC is required for mission critical applications such as factory automation, e-health, autonomous driving, tactile Internet, and augmented reality to facilitate real-time machine-to-machine and human-to-machine interaction \cite{Toward}. URLLC imposes strict quality-of-service (QoS) requirements including a very low latency (e.g., $1\,$ms) and a low packet error probability (e.g., $10^{-6})$ \cite{Toward}. In addition, the data packet size is typically small, e.g., around 160~bits \cite{Popovski1}. Existing mobile communication systems cannot meet these requirements. For example, for the long term evolution (LTE) system, the total frame time is 10$\,$ms, which exceeds the total latency requirement of URLLC applications \cite{Mehdi1}. The main challenges for the design of URLLC systems are the two contradicting requirements of low latency and ultra high reliability. For this reason, new design strategies are needed to enable URLLC. 

Modern communication systems employ multi-carrier transmission, e.g., orthogonal frequency division multiple access (OFDMA), due to its ability to exploit multi-user diversity, its robustness to multipath fading, and the flexibility it provides for the allocation of resources, such as power and bandwidth \cite{6251827}. Furthermore, multiple antenna technology  provides more degrees of freedom for resource allocation and facilitates multiplexing and diversity gains\cite{6251827}. Hence, future communication networks are expected to combine the concepts of multiple antennas, OFDMA, and URLLC. 

However, with the exception of our conference paper \cite{ghanem1}, the resource allocation algorithm design for OFDMA-URLLC systems has not been studied, yet. The authors in \cite{Seong1} studied the weighted sum rate maximization for multi-user downlink OFDMA systems. In \cite{enrgyefficient}, the authors studied the resource allocation algorithm design for energy-efficient communication in multi-cell OFDMA systems. The authors in \cite{multirelayOFDM} investigated the joint optimal power, sub-carrier, and relay node allocation in multi-relay assisted dual-hop cooperative orthogonal frequency division multiplexing (OFDM) systems. In \cite{5456049}, the authors studied the resource allocation for multiple-input single-output (MISO) OFDMA systems, where a base station (BS) equipped with multiple antennas served multiple single antenna users. However, the resource allocation algorithms proposed in \cite{Seong1,enrgyefficient, multirelayOFDM,5456049} were based on Shannon's capacity formula for the additive white Gaussian noise (AWGN) channel. Since URLLC systems employ a short frame structure and a small packet size to reduce latency, the relation between the achievable rate, decoding error probability, and transmission delay cannot be captured by Shannon's capacity formula which assumes infinite block length and zero error probability \cite{shannon}. If Shannon's capacity formula is utilized for resource allocation design for URLLC systems, the latency will be underestimated and the reliability will be overestimated, and as a result, the QoS requirements of the users cannot be met. Therefore, the results in \cite{Seong1,enrgyefficient, multirelayOFDM,5456049} and the related literature are not applicable for resource allocation in MISO OFDMA-URLLC systems. Hence, new resource algorithms for MISO OFDMA systems taking into account the specific properties and requirements of URLLC are needed, which is the main motivation for this paper.            

In recent years, the performance limits of short packet communication (SPC) \cite{thesis} have received significant attention in the literature. These performance limits provide a relationship between the achievable rate, decoding error probability, and packet length. The pioneering work in \cite{strassen} investigated the limits of SPC for discrete memoryless channels, while the authors in \cite{Polyanskiy} extended this analysis to different types of channels, including the AWGN channel and the Gilbert-Elliot channel. SPC for parallel Gaussian channels was analysed in \cite{thesis}, while in \cite{Erseghe1} an asymptotic analysis based on the Laplace integral was provided for the AWGN channel, parallel AWGN channels, and the binary symmetric channel (BSC). In \cite{Quasi}, the authors investigated the maximum achievable rate for SPC over quasi-static multiple-input multiple-output fading channels. The results in \cite{strassen,Polyanskiy,thesis,Erseghe1,Quasi} motivated the investigation of resource allocation design for SPC. In particular, optimal power allocation in a multi-user time division multiple access (TDMA) URLLC system was considered in \cite{optimal,wpspc,convexfinite}. In \cite{csunoptimizing}, the energy efficiency is maximized by optimizing the antenna configuration, bandwidth allocation, and power control under latency and reliability constraints. In \cite{chsecross},
a cross-layer framework based on the effective
bandwidth was proposed for optimal resource allocation under QoS constraints. The authors in \cite{miso} studied the joint uplink and downlink transmission design for URLLC in MISO systems. In \cite{Ultraa,Throughputcoh}, the authors studied a hybrid automatic repeat request (HARQ) scheme for URLLC systems.
  However, the above works\cite{optimal,convexfinite,chsecross,miso,csunoptimizing,Ultraa,Throughputcoh, wpspc,finitemop} assumed single carrier transmission which suffers from poor spectrum utilization and requires complex equalization at the receiver. Moreover, the optimization algorithms proposed in \cite{chsecross,miso} are based on a simplified version of the general  expression for the achievable rate of SPC \cite{Polyanskiy}. Thus, the optimal resource allocation for MISO OFDMA-URLLC systems is still an open problem.
  
In this paper, we study the resource allocation algorithm design for broadband downlink MISO OFDMA-URLLC systems, where a BS equipped with multiple antennas serves single antenna URLLC users. This paper makes the following main contributions:  
\begin{itemize}
	\item We propose a novel resource allocation algorithm design for multi-user  MISO OFDMA-URLLC systems. The resource allocation algorithm design is formulated as an optimization problem for maximization of the weighted sum throughput subject to QoS constraints for the URLLC users. The QoS constraints include the minimum number of transmitted bits, the maximum packet error probability, and the maximum time for transmission of a packet, i.e., the maximum delay\footnote{We note that the end-to-end (E2E) delay of data packet transmission comprises of various components including the transmission delay, queueing delay, propagation delay, and routing delay in the backhaul and core networks. In this work, we focus on the transmission delay, which is independent of the other components of the E2E delay.}.
	\item  The formulated optimization problem is a non-convex mixed-integer problem which is difficult to solve. However, we transform the problem into the canonical form of a monotonic optimization problem. This reformulation allows the application of the polyblock outer approximation method  to find the global optimal solution.
	\item To strike a balance between computational complexity and performance, we develop a low-complexity sub-optimal algorithm based on difference of convex programming and successive convex approximation to obtain a local optimal solution.
	\item Computer simulations show that the proposed sub-optimal algorithm closely approaches the performance of the optimal algorithm, despite its significantly lower complexity.  Furthermore, both algorithms achieve significant performance gains compared to two baseline schemes, especially if the users have heterogeneous delay requirements, as is expected for Internet-of-Things applications \cite{shannon}. Moreover, our results reveal that deploying multiple antennas is instrumental for achieving low latency and high reliability in URLLC systems.
\end{itemize}
  
   We note that this paper expands the corresponding conference version \cite{ghanem1} in several directions. First, in \cite{ghanem1}, resource allocation for single-antenna transceivers was considered, whereas in this paper, we study a system with a multiple-antenna BS. Moreover, in this paper, we derive the \textit{optimal} resource allocation policy for  MISO OFDMA-URLLC systems, whereas only a sub-optimal algorithm was provided in [1]. Finally, unlike \cite{ghanem1}, in this paper, we present extensive simulation results to illustrate the impact of the various system parameters on the performance of the proposed resource allocations algorithms.
    
The remainder of this paper is organized as follows. In
Section II, we present the considered system and channel models. In Section III, the proposed resource allocation problem is formulated. In Section IV, the optimal resource allocation algorithm is derived, whereas the low-complexity sub-optimal algorithm is provided in Section V. In Section VI, the performance of the proposed schemes is evaluated via computer simulations, and finally conclusions are drawn in Section VII.

\textit{Notation}: In this paper, lower-case letters refer to scalar numbers, while bold lower and upper case letters denote vectors and matrices, respectively. $\log_{2}(\cdot)$ is the logarithm with base 2. $\Tr{(\mathbf{A})}$ and $\Rank{(\mathbf{A})}$ denote the trace and the rank of matrix $\mathbf{A}$, respectively. $\mathbf{A}\succeq 0$  indicates that matrix $\mathbf{A}$ is positive semi-definite. $\mathbf{A}^{H}$ and ${\mathbf{A}}^{T}$ denote the Hermitian transpose and the transpose of matrix $\mathbf{A}$, respectively. $\R_{+}$ denotes the set of non-negative real numbers. $\mathbb{C}$ is the set of complex numbers. $\mathbf{I}_{N}$ is the $N \times N$ identity matrix. $\mathbb{H}_{N}$  denotes the set of all $N \times N$ Hermitian matrices. $|\cdot|$ and  $\|\cdot\|$ refer to the absolute value of a complex scalar and the Euclidean vector norm, respectively. The circularly symmetric complex Gaussian distribution with mean $\mu$ and variance $\sigma^{2}$ is denoted by $\mathcal{CN}(\mu,\sigma^{2})$, and $\sim$ stands for ``distributed as". $\mathcal{E}\{\cdot\}$ denotes statistical expectation. $\nabla_{x}f(\mathbf{x})$ denotes the gradient vector of function $f(\mathbf{x})$ and its elements are the partial derivatives of $f(\mathbf{x})$. $\mathbf{u}_{d}$ is the unit vector whose $d$-th entry is equal to $1$ and all other entries are equal to $0$. For any two vectors $\mathbf{x}$, $\mathbf{y}$ $\in$ $\R_{+}$, $\mathbf{x} \leq \mathbf{y}$ means $x_{i}\leq y_{i}$, $\forall i,$ where $x_{i}$ and $y_{i}$ are the $i$-th elements of $\mathbf{x}$ and $\mathbf{y}$, respectively.
\section{System and Channel Models}
In this section, we present the system and channel models adopted for  MISO OFDMA-URLLC in this paper.
\subsection{System Model}
We consider a single-cell downlink OFDMA system, where a BS equipped with $N_{T}$ antennas serves $K$ single-antenna URLLC users\footnote{ The URLLC users are assumed to employ a single antenna to ensure low hardware complexity.} indexed by $k =\{1,\dots,K\}$, cf. Fig.~\ref{model}(a). The frequency band is divided into $M$ orthogonal sub-carriers indexed by $m \in \{1,\dots,M\}$. We assume that a resource frame has a duration of $T_{\mathrm{f}}$ seconds, and consists of $N$ time slots\footnote{In current standards such as LTE, a typical sub-carrier bandwidth is 15~kHz which leads to an OFDM symbol duration of $T_{s}=66~\mu\textrm{s}$. Therefore, to meet a URLLC delay requirement of $1~\textrm{ms}$, $N$ has be smaller than 7. For larger sub-carrier spacing, larger values of $N$ are possible.} which are indexed by $n \in \{1,\dots,N\}$. Thereby, one OFDMA symbol spans one time slot, and in total $M \times N$ resource elements are available for assignment to the $K$ users, cf. Fig.~\ref{model}(b). We assume that the delay requirements of all users are known at the BS and only users whose delay requirements can potentially be met in the current resource block are admitted into the system. The maximum transmit power of the BS is $P_{\text{max}}$.
\subsection{Channel Model}
In this paper, we assume that the coherence time is larger than $T_{\mathrm{f}}$. Therefore, the channel gain for a given sub-carrier and a given transmit antenna remains constant for the considered $N$ time slots. The received signal at user $k$ on sub-carrier $m$ in time slot $n$ is  given as follows:
\begin{figure}
	\centering
	\scalebox{0.9}{
		\pstool{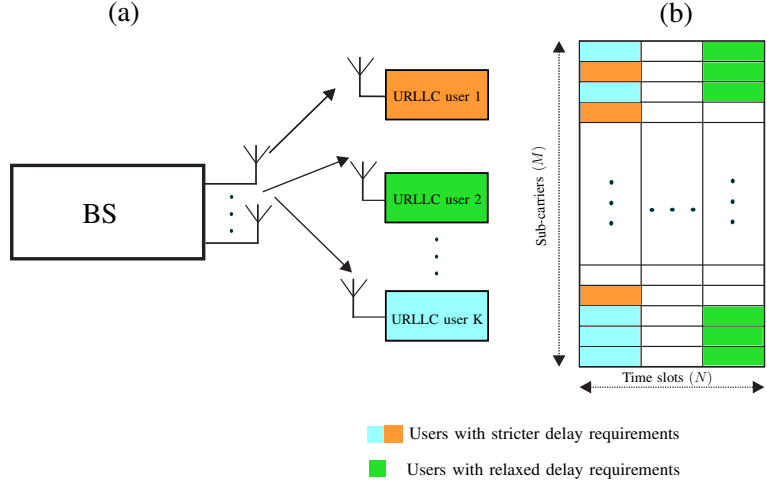}{
			\psfrag{a}[c][c][1]{\text{(a)}}
			\psfrag{b}[c][c][1]{\text{(b)}}
			\psfrag{B}[c][c][1]{\text{BS }}
			\psfrag{s}[c][c][0.5]{\text{Time slots $(N)$}}
			\psfrag{d}[c][c][0.5]{\text{Sub-carriers $(M)$}}
			\psfrag{k}[c][c][0.6]{\text{Users with stricter delay requirements}}
		   \psfrag{z}[c][c][0.6]{\text{Users with relaxed delay requirements}}
			\psfrag{e}[c][c][0.5]{\text{URLLC user 1}}
			\psfrag{f}[c][c][0.5]{\text{URLLC user 2}}
			\psfrag{c}[c][c][0.5]{\text{URLLC user K}}}}
	\caption{ Multi-user downlink  MISO OFDMA-URLLC: (a) System model with $N_{T}$-antenna BS and $K$ single-antenna users;    
	(b) Frame structure.}
\label{model}
\end{figure}
\begin{IEEEeqnarray}{lll}\label{receviedsignal}
	y_{k}[m,n]=\mathbf{h}_{k}^{H}{[m]}\mathbf{x}[m,n]+w_{k}[m,n],
\end{IEEEeqnarray}
where $\mathbf{h}_{k}[m] \in \mathbb{C}^{N_{T} \times 1}$ is the channel vector from the BS to user $k$ on sub-carrier $m$, $\mathbf{x}[m,n] \in \mathbb{C}^{N_{T} \times 1}$ is the signal vector transmitted by the BS on sub-carrier $m$ in time slot $n$. Moreover, $w_{k}[m,n]\sim \mathcal{CN}(0,\sigma^{2})$ is the complex AWGN\footnote{Without loss of generality, we assume that the noise variances for all URLLC users are identical.}. 
In this paper, we consider linear transmit precoding at the BS, where each user is assigned a beamforming vector. Hence, the transmit signal of the BS on sub-carrier $m$ in time slot $n$ is given by:
\begin{IEEEeqnarray}{lll}\label{txvector}
	\mathbf{x}[m,n]=\sum_{k=1}^{K}\mathbf{w}_{k}[m,n]u_{k}[m,n],
\end{IEEEeqnarray}
where $u_{k}[m,n] \in \mathbb{C}$ and $\mathbf{w}_{k}[m,n] \in \mathbb{C}^{N_{T} \times 1}$ are the transmit symbol and the beamforming vector of user $k$ on sub-carrier $m$ in time slot $n$, respectively. Moreover, without loss of generality, we assume that $\mathcal{E}\{|u_{k}[m,n]|^{2}\}=1, \; \forall k\in \{1,\dots, K\}$. By substituting (\ref{txvector}) into (\ref{receviedsignal}), the received signal at user $k$ on sub-carrier $m$ in time slot $n$ is given by: 
\begin{IEEEeqnarray}{lll}
	y_{k}[m,n]=\mathbf{h}_{k}^{H}[m]\left(\sum_{l=1}^{K}\mathbf{w}_{l}[m,n]u_{l}[m,n]\right)+w_{k}[m,n]\\ \ \ \quad
	=\underbrace{\mathbf{h}_{k}^{H}[m]\mathbf{w}_{k}[m,n]u_{k}[m,n]}_{\substack {\text {desired signal}}}+\underbrace{\sum_{l\neq k }\mathbf{h}^{H}_{k}[m]\mathbf{w}_{l}[m,n]u_{l}[m,n]}_{\substack {\text {multi-user interference (MUI)}}}+w_{k}[m,n].\nonumber
\end{IEEEeqnarray}
Moreover, the signal-to-interference-plus-noise-ratio (SINR) of user $k$ on sub-carrier $m$ in time slot $n$ is given as follows:
\begin{IEEEeqnarray}{lll}{\label{rho5}}
	\gamma_{k}[m,n]=\frac{|\mathbf{h}_{k}^{H}[m]\mathbf{w}_{k}[m,n]|^{2}}{\sum_{l\neq k}|\mathbf{h}_{k}^{H}[m]\mathbf{w}_{l}[m,n]|^{2}+\sigma^{2}}.
\end{IEEEeqnarray}
In this paper, we treat the interference caused by other users as noise. Moreover, to obtain a performance upper bound for MISO OFDMA-URLLC systems, for resource allocation, perfect channel state information (CSI) is assumed to be available at the BS. 
\section{Resource Allocation Problem Formulation}
In this section, we discuss the achievable rate for SPC, the QoS requirements of the URLLC users, and the adopted system performance metric for resource allocation algorithm design. Furthermore, we formulate the proposed resource allocation optimization problem for MISO OFDMA-URLLC systems. 
\subsection{Achievable Rate for SPC}
Shannon's capacity theorem, on which most conventional resource allocation designs are based, applies to the asymptotic case where the packet length approaches infinity and the decoding error probability goes to zero \cite{shannon}. Thus, it cannot be used for resource allocation design for URLLC systems, as URLLC systems have to employ short packets to achieve low latency, which also makes decoding errors unavoidable.

For performance evaluation of SPC, the so-called normal approximation for finite blocklength codes was developed in \cite{thesis}. Mathematically, the maximum number of bits $B$ conveyed in a packet comprising $L$ symbols can be approximated as\vspace*{-2mm}\cite[Eq. (4.277)]{thesis},\cite[Fig. 1]{Erseghe1}:
\begin{IEEEeqnarray}{lll}\label{normalapproximation}
	B=\sum_{i=1}^{L}\log_{2}(1+\gamma_{i})-Q^{-1}(\epsilon)\sqrt{\sum_{i=1}^{L}V_{i}},
\end{IEEEeqnarray}
where $\epsilon$ is the decoding packet error probability, $V_{i}$ is the channel dispersion, and $Q^{-1}(\cdot)$ is the inverse of the Gaussian Q-function which is given by $
Q(x)=\frac{1}{\sqrt{2\pi}}\int_{x}^{\infty}\text{exp}{\left(-\frac{t^{2}}{2}\right)}\text{d}t$. For the complex AWGN, the channel dispersion is given by \cite{thesis}\vspace*{-2mm}
\begin{IEEEeqnarray}{lll}\label{dispersion}
	V_{i}=a^{2}\bigg(1-{(1+\gamma_{i})^{-2}}\bigg),
\end{IEEEeqnarray} 
where $\gamma_{i}$ is the SINR of the $i$-th received symbol and  $a=\log_{2}(\text{e})$. 

 In this paper, we base the resource allocation algorithm design for downlink MISO OFDMA-URLLC systems on (\ref{normalapproximation}). Each resource element carries one symbol, and by allocating several resource elements from the available $L=M \times N$ resource elements to a given user, the number of bits received by the user with packet error probability $\epsilon$ can be determined based on (\ref{normalapproximation}). 
\subsection{QoS and System Performance Metric}
The QoS requirements of URLLC users include the minimum number of received bits, $B_{k}$, the target packet error probability, $\epsilon_{k}$, and the maximum number of time slots available for transmission of the user's packet,  $D_{k}$. According to (\ref{normalapproximation}), the total number of bits transmitted over the resources allocated to user $k$ can be written as:
\begin{IEEEeqnarray}{lll}\label{BT}
	\Psi_{k}(\mathbf{w}_{k})=F_{k}(\mathbf{w}_{k})-V_{k}(\mathbf{w}_{k}),
\end{IEEEeqnarray}
where
\begin{eqnarray}
F_{k}(\mathbf{w}_{k})&=&\sum_{m=1}^{M}\sum_{n=1}^{N}\log_{2}(1+\gamma_{k}[m,n]), \\
V_{k}(\mathbf{w}_{k})&=&Q^{-1}(\epsilon_{k})\sqrt{\sum_{m=1}^M \sum_{n=1}^N V_{k}[m,n]},
\end{eqnarray}
where the channel dispersion $V_{k}[m,n]$ is given by:
\begin{eqnarray}
V_{k}[m,n]=a^{2}\bigg(1-{(1+{\gamma_{k}}[m,n])}^{-2}\bigg).
\end{eqnarray}
Furthermore, $\mathbf{w}_{k}$ is the collection of all beamforming  vectors $\mathbf{w}_{k}[m,n]$, $\forall m,n,$ of user $k$.
 
The delay requirements of user $k$ can be met by assigning all symbols of user $k$ to the first $D_{k}$ time slots. In other words, users requiring low latency are assigned resource elements at the beginning of the frame, cf. Fig.~\ref{model}(b). We note that a user can start decoding as soon as it has received all OFDMA symbols that contain its data, i.e., after $D_{k}$ time slots.

In order to be able to control the fairness among the URLLC users, we adopt the weighted sum throughput as performance metric. In particular, the weighed sum throughput of the entire system is defined as:
\begin{IEEEeqnarray}{lll} \label{eq1}
	U(\mathbf{w}) =\sum_{k=1}^{K}\mu_{k}\Psi_{k}(\mathbf{w}_{k}) =F(\mathbf{w})-V(\mathbf{w}), 
\end{IEEEeqnarray}
where
\begin{IEEEeqnarray}{lll}
	\hspace*{-4mm} F(\mathbf{w})=\sum_{k=1}^{K}\mu_{k}F_{k}(\mathbf{w}_{k}), \hspace*{-2mm} \quad V(\mathbf{w})=\sum_{k=1}^{K}\mu_{k}V_{k}(\mathbf{w}_{k}),
\end{IEEEeqnarray}
and $\mu_{k}$ is the weight assigned to user $k$. Larger values of $\mu_{k}$ give a user a higher priority and, as a result, a higher throughput (i.e., more bits are transmitted to the user) compared to the other users. The value of the $\mu_{k}$ may be specified in the medium access control (MAC) layer and is assumed to be given in the following. Moreover, $\mathbf{w}$ is the collection of the beamforming vectors $\mathbf{w}_{k}$ of all users. 
\subsection{Optimization Problem Formulation}
In the following, we formulate a resource allocation optimization problem for maximization of the weighted sum throughput of the system subject to the QoS requirements of
each user regarding the received number of bits, the reliability,
and the latency. In particular, the proposed resource allocation policies are determined by solving the following optimization
problem:
\begin{IEEEeqnarray}{lll}\label{optimization1}\qquad &  \underset { {\mathbf {w}}}{ \mathop {\mathrm {maximize}}\nolimits }~ F(\mathbf{w})-V(\mathbf{w}) \\& \ \mbox {s.t.}~\mbox {C1:}\ \nonumber F_{k}(\mathbf{w}_{k})-V_{k}(\mathbf{w}_{k}) \geq B_{k}, \  \forall k, \\&\qquad \nonumber\mbox {C2:}\  \sum_{k=1}^K \sum_{m=1}^{M}\sum_{n=1}^{N}\|\mathbf{w}_{k}[m,n]\|^{2}\leq P_{\text{max}}, \\&  \nonumber\qquad \mbox {C3:}\ \mathbf{w}_{k}[m,n]=0, \quad \forall n> D_{k}, \forall k.\nonumber \qquad \end{IEEEeqnarray}
In (\ref{optimization1}), constraint $\mbox {C1}$ guarantees the transmission of a minimum number of $B_{k}$ bits to user $k$. Constraint $\mbox {C2}$ is the total power budget constraint of the BS. Finally, constraint $\mbox {C3}$ ensures that  user $k$ is served within the first $D_{k}$ time slots to meet its delay requirements. The problem in (\ref{optimization1}) is a non-convex optimization problem. The non-convexity is caused by the form of the SINR in (\ref{rho5}) and the non-convex normal approximation in (\ref{normalapproximation}) which appear in the cost function and constraint $\mbox {C1}$.
\begin{remark}
Resource allocation algorithm design for conventional, non-URLLC OFDMA systems is typically based on Shannon's capacity formula, i.e., $V(\mathbf{w})$ and $V_{k}(\mathbf{w}_{k})$ in (\ref{optimization1}) are absent \cite{Seong1,enrgyefficient, multirelayOFDM,5456049}. The presence of $V(\mathbf{w})$ and $V_{k}(\mathbf{w}_{k})$ makes problem (\ref{optimization1}) significantly more challenging to solve but is essential to capture the characteristics of OFDMA-URLLC systems.      
\end{remark}

There is no systematic approach to solving general non-convex problems optimally. However, in Section IV, we will show that based on a sequence of transformations, problem (\ref{optimization1}) can be solved optimally, by employing monotonic optimization. Moreover, in Section V, we develop a sub-optimal algorithm based on successive convex approximation and difference of convex programming to obtain close-to-optimal performance with low computational complexity.
\section{Optimal Solution of the Optimization Problem} 
In this section, we solve the optimization problem in (\ref{optimization1}) optimally based on monotonic optimization  \cite{misoyan}, which leads to an iterative resource allocation algorithm, where a semi-definite relaxation (SDR) problem is solved in each iteration. 
\subsection{Semi-Definite Programming Relaxation}
To facilitate the application of semi-definite programming (SDP), we define new variables $\mathbf{W}_{k}[m,n]=\mathbf{w}_{k}[m,n]\mathbf{w}_{k}^{H}[m,n]$ and $\mathbf{H}_{k}[m]=\mathbf{h}_{k}[m]\mathbf{h}_{k}^{H}[m]$, $\forall k,m,n$, and rewrite (\ref{optimization1}) in equivalent form as follows: 
\begin{IEEEeqnarray}{lll}\label{op2}&  \underset { {{\mathbf {W}}}}{ \mathop {\mathrm {maximize}}\nolimits }~ F({\mathbf{W}})-V({\mathbf{W}}) \\& \ \mbox {s.t.}~\mbox {C1:}\ \nonumber F_{k}({\mathbf{W}}_{k})-V_{k}({\mathbf{W}}_{k}) \geq B_{k}, \  \forall k,\qquad \\&\qquad \nonumber  \mbox {C2:}\  \sum_{k=1}^{K}\sum_{m=1}^{M} \sum_{n=1}^{N} \Tr({\mathbf{W}}_{k}[m,n])\leq P_{\text{max}}, \\&\qquad  \nonumber \mbox {C3:} \Tr({{\mathbf{W}}_{k}}[m,n]) = 0, \forall n> D_{k}, \forall k,  \nonumber\qquad \\&
\qquad	\mbox {C4:}\ {\mathbf{W}}_{k}[m,n]\succeq 0, \;\;\; \forall k,m,n, \nonumber 
	\\& \qquad \nonumber
	\mbox {C5:}\ \Rank({{\mathbf{W}}_{k}}[m,n]) \leq 1, \;\;\; \forall k,m,n,
\end{IEEEeqnarray} 
where 
\begin{IEEEeqnarray}{lll}
	F({\mathbf{W}})=\sum_{k=1}^{K}\mu_{k}\sum_{m=1}^{M}\sum_{n=1}^{N}\log_{2}\left(1+\gamma_{k}[m,n]\right), 
\end{IEEEeqnarray} 
\vspace{-0.5cm}
\begin{IEEEeqnarray}{lll}\hspace{-0.75cm}
	{V({\mathbf{W}})=\sum_{k=1}^{K}\mu_{k}aQ^{-1}(\epsilon_{k})\sqrt{\sum_{m=1}^{M}\sum_{n=1}^{N}\left(1-({1+\gamma_{k}[m,n])^{2}}\right)}}\;,
\end{IEEEeqnarray}
and
\begin{IEEEeqnarray}{lll}{\label{rho2}}
\gamma_{k}[m,n]=\frac{\Tr(\mathbf{H}_{k}[m]{\mathbf{W}}_{k}[m,n])}{\sum_{l \neq k}\Tr(\mathbf{H}_{k}[m]{\mathbf{W}}_{l}[m,n])+\sigma^{2}}.
\end{IEEEeqnarray}
We note that ${\mathbf{W}}_{k}[m,n]\succeq 0$ and $\Rank({{\mathbf{W}}_{k}}[m,n]) \leq 1,  \forall k,m,n,$ in constraints $\mbox{C4}$ and $\mbox{C5}$ are imposed to ensure that  $\mathbf{W}_{k}[m,n]=\mathbf{w}_{k}[m,n]\mathbf{w}_{k}^{H}[m,n]$ holds after optimization. Moreover, for simplicity of notation, we define ${\mathbf {W}_{k}}$ as the collection of all optimization variables $\mathbf{W}_{k}[m,n]$,~$\forall m,n$, and ${\mathbf{W}}$  as the collection of all ${\mathbf {W}_{k}},\forall k$.  
\subsection{Problem Transformation}
The objective function and constraint $\mbox{C1}$ in (\ref{op2}) have a complicated structure. To handle this complexity and to facilitate the application of monotonic optimization, we introduce a set of auxiliary variables  $z_{k}[m,n],\forall k,m,n$,  to bound the SINR from below, i.e.,
\begin{IEEEeqnarray}{lll}{\label{rho3}}
 0  \leq z_{k}[m,n]\leq \gamma_{k}[m,n] =\frac{f_{k}[m,n]({\mathbf {W}})}{g_{k}[m,n]({\mathbf {W}})}, \forall k,m,n,
\end{IEEEeqnarray} 
where $	f_{k}[m,n]({\mathbf {W}})$ and $g_{k}[m,n]({\mathbf {W}})$ are the numerator and denominator of the SINR in (\ref{rho2}) and are given respectively by
\begin{IEEEeqnarray}{lll}
	f_{k}[m,n]({\mathbf {W}})=\Tr(\mathbf{H}_{k}[m]{\mathbf{W}}_{k}[m,n]), \forall k,m,n,
\end{IEEEeqnarray}
\vspace{-1cm}
\begin{IEEEeqnarray}{lll}
	\qquad \qquad	g_{k}[m,n]({\mathbf {W}})=\sum_{l \neq k}\Tr(\mathbf{H}_{k}[m]{\mathbf{W}}_{l}[m,n])+\sigma^{2}, \forall k,m,n.
\end{IEEEeqnarray}

Let us replace $\gamma_{k}[m,n]$ by $z_{k}[m,n]$ in $F(\mathbf{W})$, $V(\mathbf{W})$, $F_{k}(\mathbf{W}_{k})$, and $V_{k}(\mathbf{W}_{k})$ and denote the resulting functions by  $F(\mathbf{z})$, $V(\mathbf{z})$, $F_{k}(\mathbf{z}_{k})$, and $V_{k}(\mathbf{z}_{k})$, respectively, i.e.,
\begin{align} 
	F({\mathbf{z}})=\sum_{k=1}^{K}\mu_{k}F_{k}(\mathbf{z}_{k}), 
\end{align}
\vspace{-0.7cm}
\begin{align} 
	V({\mathbf{z}})=\sum_{k=1}^{K}\mu_{k}V(\mathbf{z}_{k}),
\end{align}
\vspace{-0.7cm}
\begin{align}
\hspace{-2.2cm}	F_{k}(\mathbf{z}_{k})=\sum_{m=1}^{M}\sum_{n=1}^{N}\log_{2}(1+z_{k}[m,n]), \forall k,
\end{align} 
\vspace{-0.7cm}
\begin{align}
	V(\mathbf{z}_{k})=aQ^{-1}(\epsilon_{k})\sqrt{\sum_{m=1}^{M}\sum_{n=1}^{N}(1-(1+z_{k}[m,n])^{-2})}\;,
\end{align} 
where $\mathbf{z}_{k}$ denotes the collection of optimization variables $z_{k}[m,n], \ \forall m,n$, and $\mathbf{z}$ denotes the collection of optimization variables $\mathbf{z}_{k}, \forall k$. Using these notations, and after dropping rank constraint $\mbox{C5}$ in (\ref{op2}), we formulate a new optimization problem as follows:     	  
\begin{IEEEeqnarray}{lll}\label{opequivalnt}&  \underset { {{\mathbf {W}}}, \mathbf {z}}{ \mathop {\mathrm {maximize}}\nolimits }~ F({\mathbf{z}})-V({\mathbf{z}}) \\& \ \mbox {s.t.}~\mbox {C1:}\ \nonumber F_{k}({\mathbf{z}}_{k})-V_{k}({\mathbf{z}}_{k}) \geq B_{k},  \forall k, \\& \qquad \nonumber \mbox {C2-C4,}
		\\& \nonumber \qquad \mbox {C6:} \; z_{k}[m,n]\leq \frac{f_{k}[m,n]({\mathbf {W}})}{g_{k}[m,n]({\mathbf {W}})}, \forall k,m,n, \;
		\\& \nonumber \qquad \mbox {C7:} \; z_{k}[m,n] \geq 0. \;
\end{IEEEeqnarray}  

In the following, we first find an optimal solution for problem (\ref{opequivalnt}). Subsequently, we prove that problems (\ref{opequivalnt}) and (\ref{op2})  are equivalent, cf. \textbf{Proposition} \ref{pro1}. Hence, the solution obtained for problem (\ref{opequivalnt}) constitutes an optimal solution for problem (\ref{op2}), too. 
\color{black}

The main condition required for applying monotonic optimization is the monotonicity of the objective function and the constraints. We note that the objective function and constraint $\mbox {C1}$ in (\ref{opequivalnt}) are differences of two monotonic concave functions in the optimization variables $\mathbf{z}$, cf. Appendix~A. Hence, problem (\ref{opequivalnt}) can be transformed into the canonical form of a monotonic optimization problem in two steps:    
\begin{itemize}
	\item \textbf{Step 1:} To transform the objective function in (\ref{opequivalnt}) into a monotonic function, we note that the SINR in (\ref{rho3}) is upper bounded by ${z}_{\text{max},k}[m,n]$ \footnote{The right hand side of (\ref{ineq}) is obtained by allocating all available power $P_{\text{max}}$ to time slot $n$, sub-carrier $m$, and user $k$.}:
	\begin{IEEEeqnarray}{lll}\label{ineq}
	 z_{k}[m,n] \leq {z}_{\text{max},k}[m,n]\triangleq \frac{P_{\text{max}}}{\sigma^{2}}\Tr({\mathbf{H}_{k}[m,n]}), \forall k,m,n.
	\end{IEEEeqnarray}

	 Let us define $\mathbf{z}_{\text{max}}$ as the collection of all ${z}_{\text{max},k}[m,n]$. Since  $V(\mathbf{z})$ is monotonically increasing in $\mathbf{z}$, $\mathbf{z} \leq \mathbf{z}_{\text{max}}$ leads to $V(\mathbf{z}) \leq V(\mathbf{z}_{\text{max}})$. Therefore, $V(\mathbf{z})+t = V(\mathbf{z}_{\text{max}})$ holds, for some positive $t$. Hence, substituting $V(\mathbf{z})$ by $V(\mathbf{z}_{\text{max}})-t$, the optimization problem in (\ref{opequivalnt}) can be rewritten as follows:   	
	\begin{IEEEeqnarray}{lll}\label{optimization4bb1}& \underset {{{\mathbf {W}}}, \mathbf{z}, t}{ \mathop {\mathrm {maximize}}\nolimits }~ F({\mathbf{z}})+t-V(\mathbf{z}_{\text{max}}) \\& \ \mbox {s.t.}~ \nonumber \ \mbox{C1-C4, C6, C7},
		\\& \qquad ~\mbox {C8:}\ \nonumber t +V({\mathbf{z}}) \leq V({\mathbf{z}_{\text{max}}}), \\& \qquad ~\mbox {C9:}\ \nonumber t \geq 0.
	\end{IEEEeqnarray}

We note that at the optimal point constraint $\mbox{C8}$ holds with equality due to the monotonicity of the objective function with respect to auxiliary optimization variable $t$. 

\item \textbf{Step 2:} We use a similar approach as for transforming the cost function to transform  constraint $\mbox{C1}$ into a standard monotonic constraint. In particular, $V_{k}({\mathbf{z}}_{k}) +\zeta_{k} = V_{k}(\mathbf{z}_{{\text{max}},k})$ holds for some positive auxiliary optimization variable $\zeta_{k}$, where $\mathbf{z}_{{\text{max}},k}$ is the collection of the ${z}_{{\text{max}},k}, \forall m,n$. Therefore, by substituting $V_{k}({\mathbf{z}}_{k})$ by $  V_{k}(\mathbf{z}_{{\text{max},k}}) -\zeta_{k}$, constraint $\mbox{C1}$ can be transformed into two monotonic constraints as follows:
	\begin{IEEEeqnarray}{lll}\label{optimization4bs1}
	\; \; \qquad	\mbox {C1a:} \quad F_{k}({\mathbf{z}}_{k}) + \zeta_{k} \geq V_{k}(\mathbf{z}_{{\text{max},k}})+B_{k},\forall k, 
	\end{IEEEeqnarray}
\vspace{-1cm}
	\begin{IEEEeqnarray}{lll}\label{optimization4bs2}
		\mbox {C1b:} \quad V_{k}({\mathbf{z}}_{k}) +\zeta_{k} \leq V_{k}(\mathbf{z}_{{\text{max},k}}),\forall k.
	\end{IEEEeqnarray}
\end{itemize}
We note that the left hand sides of (\ref{optimization4bs1}) and (\ref{optimization4bs2}) are monotonically increasing functions. Hence, problem (\ref{optimization4bb1}) has been transformed to an equivlant  monotonic optimization problem as follows:
\begin{IEEEeqnarray}{lll}\label{optimization22b}& \underset {{{\mathbf {W}}},\mathbf{z},{t}, \boldsymbol{{\zeta}}}{ \mathop {\mathrm {maximize}}\nolimits }~ F({\mathbf{z}})+t\\&  \ \mbox {s.t.}~  \nonumber \mbox{C1a, C1b, C2-C4, C6-C9}, 
\end{IEEEeqnarray}
where $\boldsymbol{\zeta}$ is the collection of optimization variables $\zeta_{k},\ \forall k$. Note that, in (\ref{optimization22b}),  we removed the constant $V(\mathbf{z}_{\text{max}})$ from the objective function, because it has no effect on the optimal solution. Optimization problem (\ref{optimization22b}) has a monotonically increasing objective function and all constraints are monotonically increasing functions ($\mbox{C1b, C6, C8}$) or convex functions 
($\mbox{C1a, C2-C4, C7, C9}$). Therefore, (\ref{optimization22b}) belongs to the class of monotonic optimization problems\cite{yanmisoc,Tuy}, which can be solved using algorithms such as outer polyblock approximation. To facilitate the presentation of the proposed solution, we rewrite the problem (\ref{optimization22b}) in the canonical form of a monotonic optimization problem as follows:
\begin{IEEEeqnarray}{lll}\label{optimization5bb}
	\underset {{{\mathbf {W}}}, \mathbf{z}, {t},  \boldsymbol{\zeta}}{ \mathop {\mathrm {maximize}}\nolimits }~ F({\mathbf{z}})+t\\ \ \mbox {s.t.}~ ({{\mathbf {W}}},\mathbf{z},t,\mathbf{\zeta}) \in \mathcal{V},\nonumber
\end{IEEEeqnarray}
where the feasible set $\mathcal{V}=\mathcal{G} \cap \mathcal{H}$ is the intersection of the normal set $\mathcal{G}$ and the co-normal set $\mathcal{H}$ \cite{Zhangmonotonic}. The normal set $\mathcal{G}$ is given by:
\begin{IEEEeqnarray}{lll}\label{normalset}
	\mathcal{G}=\left\{{(\mathbf{z},t)|0 \leq z_{k}[m,n] \leq \frac{f_{k}[m,n]({\mathbf {W}})}{g_{k}[m,n]({\mathbf {W}})}, \forall k, m, n, \ \mathbf{z} \in \mathcal{Z},   \mathbf{W} \in \mathcal{W}} 
\right\},\end{IEEEeqnarray}
with $\mathcal{Z}$ and $\mathcal{W}$ being the feasible set spanned by constraints ($\mbox{C1b, C2-C4, C6-C9}$). The co-normal set $\mathcal{H}$ is defined by constraint $\mbox {C1a}$. Now, we are ready to design the optimal resource allocation algorithm based on the polyblock outer approximation algorithm \cite{yanmisoc}.
\subsection{Polyblock Algorithm}
Due to the monotonicity of the objective function and the constraints, the optimal solution of  optimization problem (\ref{optimization5bb}) is at the boundary of the feasible set $\mathcal{V}$ \cite{Tuymonotonic,emiloptimal,Zhangmonotonic}. However, the boundary of the feasible set is unknown. Thus, we approach the boundary from above by enclosing the feasible set  $\mathcal{V}$ by an initial polyblock $\mathcal{B}^{(1)}$ with an initial vertex $\mathbf{v}^{(1)}=\left( \mathbf{z}_{1}^{(1)},t_{1}^{(1)}\right)$ as shown in Fig.~\ref{polyblockfig}(a), where for simplicity of illustration, we consider a case with only two dimensions $t_{1}$ and $z_{1}$ to depict the polyblock algorithm. Subsequently, the intersection point $\Phi(\cdot)$ between the vertex and the origin is calculated, and the new polyblock $\mathcal{B}^{(2)}$ is now defined by three vertices $\mathbf{v}^{(1)},\tilde{\mathbf{v}}^{(1)}$, and $\tilde{\mathbf{v}}^{(2)}$, see Fig.~\ref{polyblockfig}(b). Since, vertex $\mathbf{v}^{(1)}$ has no effect on polyblock $\mathcal{B}^{(2)}$, we can remove it, see Fig.~\ref{polyblockfig}(c). This process is continued until the feasible set $\mathcal{V}$ is enclosed by a final polyblock $\mathcal{B}^{(1)} \supset \mathcal{B}^{(2)} \supset \dots \supset \mathcal{V}$. Finally, we select the vertex that maximizes the objective function in (\ref{optimization5bb}). This procedures is summarized in \textbf{Algorithm} \ref{polyblocka}. \textbf{Algorithm} \ref{polyblocka} requires the calculation of intersection point $\Phi(\cdot)$ in each iteration, which is performed by \textbf{Algorithm} \ref{intersectiona}, as explained in the following.

\begin{figure}
	\centering
	\scalebox{0.5}{
		\pstool{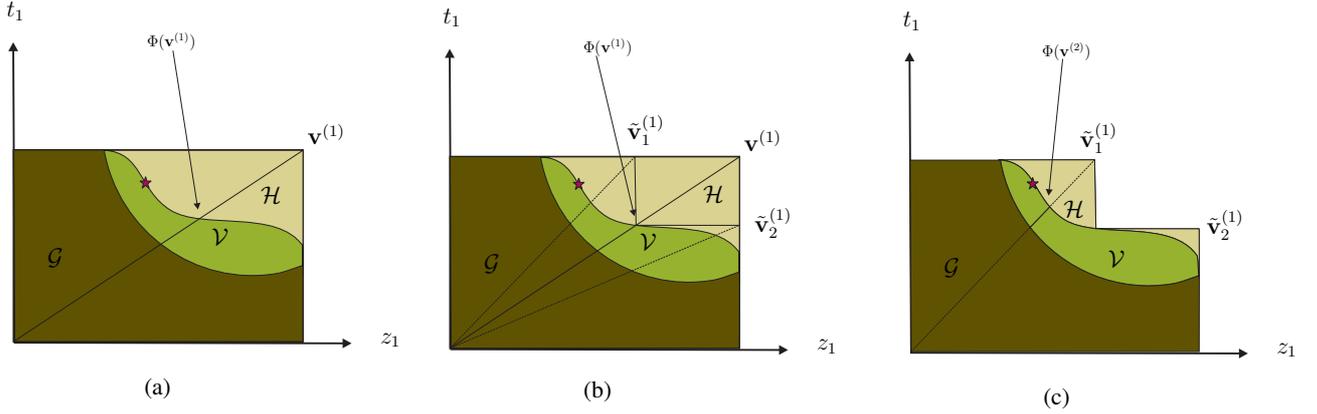}{
			\psfrag{t2}[c][c][1.5]{$t_{1}$}
			\psfrag{z1}[c][c][1.5]{$z_{1}$}
		\psfrag{a}[c][c][1.5]{\text{(a)}}
		\psfrag{b}[c][c][1.5]{\text{(b)}}
			\psfrag{c}[c][c][1.5]{\text{(c)}}
			\psfrag{g}[c][c][1.5]{$\mathcal{G}$}
			\psfrag{v}[c][c][1.5]{$\mathcal{V}$}
			\psfrag{h}[c][c][1.5]{$\mathcal{H}$}
			\psfrag{w}[c][c][1.5]{$\mathbf{v}^{(1)}$}
			\psfrag{d}[c][c][1.5]{$\tilde{\mathbf{v}}^{(1)}_{1}$}
		\psfrag{z}[c][c][1.5]{$\tilde{\mathbf{v}}^{(1)}_{2}$}
			\psfrag{m}[c][c][1.5]{$\tilde{\mathbf{v}}^{(2)}_{2}$}
			\psfrag{m}[c][c][1.5]{$\mathbf{v}^{(2)}$}
			\psfrag{I}[c][c][1]{$\Phi(\mathbf{v}^{(1)})$}
		\psfrag{o}[c][c][1]{$\Phi(\mathbf{v}^{(2)})$}}}
	\caption{The polyblock outer approximation algorithm. $\mathcal{G}$ is the normal set, $\mathcal{H}$ is the co-normal set, and $\mathcal{V}=\mathcal{G} \cap \mathcal{H} $. $\Phi(\cdot)$ is the intersection point between the boundary and the line connected a vertex with the origin. The star is the optimal solution located at the boundary of the feasible set $\mathcal{V}$.}
		\label{polyblockfig}
\end{figure}
\begin{figure}
\end{figure}
\begin{algorithm}[t]
	\caption{Polyblock Outer Approximation Algorithm}
	\begin{algorithmic}[1]
		\renewcommand{\algorithmicrequire}{\textbf{Input:}}
		\renewcommand{\algorithmicensure}{\textbf{Output:}}
		\STATE Initialize polyblock $\mathcal{B}^{(1)}$ with vertex set $\mathbf{v}^{(1)}$ =$\{\mathbf{z}^{(1)},t^{(1)}\}$, where the elements of $\mathbf{z}^{(1)}$ are   $z_{k}^{(1)}[m,n]=\Tr(\mathbf{H}_{k}[m])P_{\text{max}}, \forall k,m,n,$ and $t^{(1)}=V(\mathbf{z}_{\text{max}})$,
		\STATE Set error tolerance $\rho \ll 1$ and iteration index $j=1$.
		\STATE	\textbf{Repeat}\{Main Loop\}\\
		\STATE Construct a smaller polyblock $\mathcal{B}^{(j+1)}$ with vertex set $\mathbf{v}^{(j+1)}$ by replacing  $\mathbf{v}^{(j)}$ with $D=K\times M \times N+1$ new vertices \{$\tilde{\mathbf{v}}_{1}^{(j)},\dots,\tilde{\mathbf{v}}_{D}^{(j)}$\} , $d \in \ \{1,\dots,D\}$. The new vertex $\tilde{\mathbf{v}}_{d}^{(j)}$ is generated as
			\begin{equation*}
			\tilde{\mathbf{v}}_{d}^{(j)}=\mathbf{v}^{(j)}-\left({v}_{d}^{(j)}-\phi_{d}(\mathbf{v}^{(j)}) \right)\mathbf{u}_{d},
			\end{equation*}
			where ${v}_{d}^{(j)}$ and $\phi_{d}(\mathbf{v}^{(j)})$ are the $d$-th elements of $\mathbf{v}^{(j)}$ and $\Phi(\mathbf{v}^{(j)})$, respectively. $\Phi(\mathbf{v}^{(j)})$ is obtained by \textbf{Algorithm} \ref{intersectiona}.
			\STATE Find $\mathbf{v}^{(j+1)}$ as that vertex of $\mathcal{V}^{(j+1)}$ whose intersection maximizes the objective function of problem 
				\begin{equation*}
				\mathbf{v}^{(j+1)}=\underset{{\mathbf{v} \in \mathcal{V}^{(j+1)}}}{\text{argmax}}\{F({\mathbf{z}})+t\},
				\end{equation*} 
		         \STATE Set $j=j+1$\\
		         \STATE \textbf{until}  $\frac{\|\mathbf{v}^{(j)}-\Phi(\mathbf{v}^{(j)})\|}{\|\mathbf{v}^{(j)}\|} \leq \rho$\\
		         \STATE Optimal vertex $\mathbf{v}^{*}=\Phi(\mathbf{v}^{(j)})$ and optimal beamforming matrix $\mathbf{W}^{*}$ are obtained when calculating $\Phi(\mathbf{v}^{(j)})$ with \textbf{Algorithm} \ref{intersectiona}.
	          	\end{algorithmic} 
		         	\label{polyblocka}
\end{algorithm}
\begin{algorithm}[t]
	\caption{Optimal Intersection Algorithm via Bisection Method}
	\begin{algorithmic}[1]
		\renewcommand{\algorithmicrequire}{\textbf{Input:}}
		\STATE {Initialize} feasible set $\mathcal{V
		}$, vertex $\mathbf{v}^{(j)}=\{\mathbf{z}^{(j)},t^{(j)}\}$.
		 $\lambda_{\text{min}}=0$ and $\lambda_{\text{max}}=1$ 
		 \STATE Set error tolerance $\delta\ll 1$.
		\WHILE{$(\lambda_{\text{max}} -\lambda_{\text{min}}) \geq \delta$}
		\STATE  ${\lambda}=(\lambda_{\text{max}} +\lambda_{\text{max}} )/2$.
		\STATE Check the feasibility of problem (\ref{intersection}) using, e.g., CVX, and check if  ${\lambda}t^{(j)} + V({\lambda}{\mathbf{\mathbf{z}}^{(j)}}) \leq V(\mathbf{z}_{\text{max}}),  0 \leq  {\lambda} t^{(j)} \leq V(\mathbf{z}_{\text{max}})$.\label{line5}
		\IF{ the two conditions in line $\ref{line5}$ are satisfied} 
		\STATE   set  ${\lambda}=\lambda_{\text{min}}$
		\ELSE
		\STATE set ${\lambda}=\lambda_{\text{max}}$
		\ENDIF
		\ENDWHILE \\
		\STATE ${\lambda}=\lambda_{\text{min}}$,   $\Phi(\mathbf{v}^{(j)})={\lambda}{\mathbf{v}}^{(j)}.$
	\end{algorithmic} 
\label{intersectiona}
\end{algorithm}
\subsection{Calculation of Intersection Point}  
In each iteration of \textbf{Algorithm} \ref{polyblocka}, we determine the intersection of the line connecting $\mathbf{v}^{(j)}$ and the origin with the boundary of the feasible set, cf. Fig.~\ref{polyblockfig}(a). In other words, we have to find a $\lambda > 0$ which satisfies $\Phi(\mathbf{v}^{(j)})=\lambda \mathbf{v}^{(j)}$. $\lambda$ can be obtained based on the following optimization problem:
\begin{IEEEeqnarray}{lll}\label{i1}				
		{ \mathop {\mathrm {maximize}}\nolimits } \qquad \lambda \\  
		 \mbox {s.t.}~ \quad \lambda \mathbf{v}^{(j)} \in \mathcal{V}\nonumber.	
\end{IEEEeqnarray}
To solve (\ref{i1}), we use the bisection method which is formally presented in \textbf{Algorithm} \ref{intersectiona}. In particular, in line $\ref{line5}$ of \textbf{Algorithm} \ref{intersectiona}, we solve the following SDP problem: 
\begin{IEEEeqnarray}{lll}\label{intersection}
	\underset {{{\mathbf {W}}}, \boldsymbol{\zeta}}{ \mathop {\mathrm {maximize}}\nolimits }~ 1\\ \mbox {s.t.}~\mbox {C1a:}\; F_{k}({\lambda}{\mathbf{z}}^{(j)}_{k}) + \zeta_{k} \geq V_{k}({\mathbf{z}_{\text{max},k}})+B_{k}, \forall k, \nonumber
	\\ \;\;\quad \mbox {C1b:} \; V_{k}({\lambda}{\mathbf{z}}^{(j)}_{k}) +\zeta_{k} \leq V_{k}({\mathbf{z}_{\text{max},k}}),  \forall k, \nonumber
	\\ \;\;\quad\nonumber  \mbox {C2-C4, C7},\\ \;\; \quad
	\mbox {C6:}\  ({\lambda}{z}^{(j)}_{k}[m,n]) g_{k}[m,n]({\mathbf {W}}) -f_{k}[m,n]({\mathbf {W}}) \leq 0,   \forall k,m,n\nonumber.
\end{IEEEeqnarray}
 Problem (\ref{intersection}) is a convex optimization problem which can be solved using standard optimization software tools such as CVX \cite{cvx}. Moreover, the tightness of the applied SDR is revealed in the following theorem.

\begin{thm}\label{th1}
	The optimal $\mathbf{W}_{k}[m,n], \forall k,m,n,$ as the solution of (\ref{intersection}) has a rank less than or equal to one, i.e.,  $\Rank(\mathbf{W}_{k}[m,n]) \leq 1, \forall k,m,n$.
\end{thm}
\begin{proof}
	Please refer to Appendix~B.
\end{proof}

\begin{prop}\label{pro1}
	Optimization problems (\ref{opequivalnt}) and (\ref{op2}) are equivalent in the sense that they yield the same solution for the beamforming matrix $\mathbf{W}_{k}[m,n], \forall k,m,n$.
\end{prop}
\begin{proof}
The solution of (\ref{opequivalnt}) is the same as that of (\ref{op2}) if 
i) constraint $\mbox{C6}$ in (\ref{opequivalnt}) holds with equality and ii) $\mathbf{W}_{k}[m,n], \forall k,m,n,$ obtained from (\ref{opequivalnt}) has rank smaller than or equal to one. Problem (\ref{opequivalnt}) is solved with \textbf{Algorithm} \ref{polyblocka} where in each iteration problem  (\ref{intersection}) is solved. In Theorem 1, we showed that the $\mathbf{W}_{k}[m,n], \forall k,m,n,$ obtained from (\ref{intersection}) have rank equal to or smaller than one. This implies that the solution of (\ref{opequivalnt}) has also  rank equal to or smaller than one, i.e., condition ii) holds. Moreover, in Section IV-B, we showed that (\ref{opequivalnt}) is a monotonic optimization problem. This implies that the optimal solution lies on the boundary of the feasible set of (\ref{opequivalnt}). As a consequence, constraint $\mbox{C6}$ in (\ref{opequivalnt}) has to hold with equality, i.e., $z_{k}[m, n]=\gamma_{k}[m, n], \forall k,m,n$. Hence, condition i) is also satisfied. This completes the proof.   
%
%
\end{proof}
\color{black}
The computational complexity  of the optimal scheme is exponential in the number of vertices, $D$, used in each iteration. Nevertheless, the obtained global solution constitutes a valuable performance upper bound for any sub-optimal resource allocation algorithm. In the next section, we propose a sub-optimal resource allocation algorithm which has polynomial time computational complexity and yields close-to-optimal performance. 
\section{Low-complexity Resource Allocation Algorithm} 
In this section, we propose a low-complexity resource allocation algorithm based on penalized successive convex approximation providing a locally optimal solution of optimization problem (\ref{optimization1}).
\subsection{Difference of Convex Programming}
In this subsection, we solve optimization problem (\ref{opequivalnt}), as (\ref{opequivalnt}) is equivalent to (\ref{optimization1}). We solve optimization problem (\ref{opequivalnt}) in two steps. First, we transform the problem into the canonical form needed for application of difference of convex programming. Second, we apply a Taylor series expansion to obtain a convex approximation of the non-convex terms. As a result, we obtain a convex optimization problem that can be efficiently solved using convex optimization software. In the following, we explain these two steps in detail.

\textbf{Step 1:} We note that non-convex constraint $\mbox{C6}$ in (\ref{opequivalnt}) can be rewritten as follows: 
\begin{IEEEeqnarray}{lll}\label{prodd1}
	\mbox {C6:} \; z_{k}[m,n]{g_{k}[m,n]({\mathbf {W}})}=z_{k}[m,n](I_{k}[m,n]({\mathbf {W}})+\sigma^{2})\leq {f_{k}[m,n]({\mathbf {W}})}, \; \forall k,m,n, \; 
\end{IEEEeqnarray} 
where $g_{k}[m,n]=I_{k}[m,n]({\mathbf {W}})+\sigma^{2}$. 
We note that $z_{k}[m,n]I_{k}[m,n]({\mathbf {W}})$ in (\ref{prodd1}) is a bilinear term which is non-convex. In fact, the Hessian matrix of a bilinear function is neither positive nor negative semi-definite. Thus, bilinear functions are neither convex nor concave in general, which is an obstacle for designing computationally
efficient resource allocation algorithms. The product of two convex function $f_{1}(x)$ and $f_{2}(x)$ can be written as a difference of two convex functions as follows\cite{Tuybookgo}:
\begin{IEEEeqnarray}{lll}\label{dcr1}
	f_{1}(x)f_{2}(x)=0.5(f_{1}(x)+f_{2}(x))^{2}-0.5f_{1}(x)^{2}-0.5f_{2}(x)^{2}.
\end{IEEEeqnarray} 

Exploiting (\ref{dcr1}) with $z_{k}[m,n]$ and $I_{k}[m,n]({\mathbf {W}})$ as $f_{1}$ and $f_{2}$, respectively, we can express the product term $z_{k}[m,n]I_{k}[m,n]({\mathbf {W}})$ in (\ref{prodd1}) as follows:
\begin{IEEEeqnarray}{lll}\label{dcr2}
	\hspace{-0.5cm}	z_{k}[m,n]I_{k}[m,n]({\mathbf {W}})=Q(z_{k}[m,n],\mathbf{W})-T(z_{k}[m,n],\mathbf{W}),
\end{IEEEeqnarray} 
where
\begin{IEEEeqnarray}{lll}
	\hspace{-0.5cm}	Q(z_{k}[m,n],\mathbf{W})=\frac{1}{2}(z_{k}[m,n]+I_{k}[m,n]({\mathbf {W}}))^{2},\forall k,m,n,\\ \hspace{-0.5cm} T(z_{k}[m,n],\mathbf{W})=\frac{1}{2}(z_{k}[m,n])^{2}+\frac{1}{2}(I_{k}[m,n]({\mathbf {W}}))^{2}, \forall k,m,n.
\end{IEEEeqnarray}
 
Furthermore, substituting (\ref{dcr2}) into (\ref{prodd1}), we obtain an equivalent representation for constraint $\mbox{C6}$ in (\ref{prodd1}) as follows:
\begin{IEEEeqnarray}{lll}\label{c6a}
	\mbox {C6:} \; Q(z_{k}[m,n],\mathbf{W})-T(z_{k}[m,n],\mathbf{W}) \leq  {f_{k}[m,n]({\mathbf {W}})}-\sigma^{2}z_{k}[m,n],\forall k,m,n, \; 
\end{IEEEeqnarray} 
 where the left hand side is a difference of two convex functions. Hence, optimization problem (\ref{opequivalnt}) can now be rewritten as follows:  
\begin{IEEEeqnarray}{lll}\label{eq2}&  \underset { {{\mathbf {W}}}, \mathbf {z}}{ \mathop {\mathrm {minimize}}\nolimits }~ -[F({\mathbf{z}})-V({\mathbf{z}})] \\& \ \mbox {s.t.}~\mbox {C1-C4, C7,} \nonumber
	\\& \nonumber \qquad \mbox {C6:}\; Q({z}_{k}[m,n],\mathbf{W})-T({z}_{k}[m,n],\mathbf{W}) \leq  {f_{k}[m,n]({\mathbf {W}})}-\sigma^{2}z_{k}[m,n],\forall k,m,n.
\end{IEEEeqnarray} 

The optimization problem in (\ref{eq2}) belongs to the class of difference of convex programming problems, since its objective function can be written as a difference of two convex functions and constraints $\mbox {C1}$ and $\mbox {C6}$ can also be expressed as the differences of two convex functions. In particular, functions $-F({\mathbf{z}})$, $-V({\mathbf{z}})$, $Q({z}_{k}[m,n],\mathbf{W})$, and $T({z}_{k}[m,n],\mathbf{W})$ are convex functions.

\textbf{Step 2:}  To obtain a convex optimization problem that can be efficiently solved, we have to handle the non-convex objective function and non-convex constraints $\mbox {C1}$ and $\mbox {C6}$. To this end, we determine the first order approximations of functions ${V}_{k}({\mathbf{z}}_{k})$ and  $T(z_{k}[m,n],\mathbf{W})$ using Taylor series as follows:
\begin{eqnarray}
\label{inequalit1b}
{V}_{k}(\mathbf{z}_{k}) &\leq \bar{V}_{k}(\mathbf{z}_{k}) = & {V}(\mathbf{z}_{k}^{({j})})+ \nabla_{\mathbf{z}_{k}}{V}_{k}(\mathbf{z}_{k}^{({j})})^{T}(\mathbf{z}_{k}-\mathbf{z}_{k}^{({j})}),
\end{eqnarray}
and
\begin{multline} \label{inequalit1c}
T({z}_{k}[m,n],\mathbf{W}) \geq \bar{T}({z}_{k}[m,n],\mathbf{W})=  T({z}^{(j)}_{k}[m,n],\mathbf{W}^{(j)})+\\ \nabla_{{z}_{k}[m,n]}{T}(z^{({j})}_{k}[m,n],\mathbf{W}^{(j)})({z}_{k}[m,n]-{z}^{({j})}_{k}[m,n])\\+\Tr(\nabla_{\mathbf{W}}{T}({z}^{(j)}_{k}[m,n],\mathbf{W}^{(j)})^{T})(\mathbf{W}-\mathbf{W}^{({j})}), \forall k,m,n,
\end{multline}
where $\mathbf{W}^{({j})}$, $\mathbf{z}_{k}^{({j})}$, and ${z}^{({j})}_{k}[m,n]$ are initial feasible points, and
\begin{eqnarray} \nabla_{\mathbf{z}_{k}}\mathbf{V}_{k}(\mathbf{z}_{k})
= \frac{a^{2} 
	Q^{-1}(\epsilon_{k})}{\sqrt{\sum_{m=1}^{M}\sum_{n=1}^{N} {V}_{k}[m,n]}}\begin{pmatrix} 
\ \frac{1}{(1+z_{k}[1,1])^{3}} \\
\ \frac{1}{(1+z_{k}[2,1])^{3}} \\
\vdots \\
\frac{1}{(1+z_{k}[M,N])^{3}} \end{pmatrix},
\end{eqnarray}
\begin{eqnarray}
\nabla_{{z}_{k}[m,n]}{T}({z}_{k}[m,n],\mathbf{W})={z}_{k}[m,n],
\end{eqnarray}
and
\begin{eqnarray}
\nabla_{\mathbf{W}}{T}({z}_{k}[m,n],\mathbf{W})=I_{k}[m,n](\mathbf{W})\mathbf{H}_{k}[m].
\end{eqnarray}
The right hand sides of (\ref{inequalit1b}) and (\ref{inequalit1c}) are affine functions, and by substituting them in (\ref{eq2}), we obtain the following convex optimization problem:  
\begin{IEEEeqnarray}{lll}\label{eq3}&  \underset { {{\mathbf {W}}}, \mathbf {z}}{ \mathop {\mathrm {minimize}}\nolimits }~ -[F({\mathbf{z}})-\bar{V}({\mathbf{z}})] \\& \ \mbox {s.t.}~\mbox {C1:}\;\nonumber F_{k}({\mathbf{z}}_{k})-\bar{V}_{k}({\mathbf{z}}_{k}) \geq B_{k}, \  \forall k, \\& \qquad \nonumber \mbox {C2-\;C4, C7},
	\\& \nonumber \qquad \mbox {C6:}\;Q(z_{k}[m,n],\mathbf{W})-\bar{T}(z_{k}[m,n],\mathbf{W}) \leq  {f_{k}[m,n]({\mathbf {W}})}-z_{k}[m,n],\forall k,m,n.
\end{IEEEeqnarray} 

Optimization problem (\ref{eq3}) can be efficiently solved by standard convex solvers such as CVX \cite{cvx}. Problem (\ref{eq3}) can be solved iteratively where the solution of (\ref{eq3}) in iteration ${j}$ is used as the initial point for the next iteration ${j}+1$. The algorithm produces a sequence of improved feasible solutions until convergence to a local optimum point of problem (\ref{eq3}) or equivalently problem (\ref{optimization1}) in polynomial time \cite{yan,Joinoptimization}.  Moreover, one can show that the solution to (\ref{eq3}) yields a matrix that has a rank equal to or smaller than one, i.e., $\Rank(\mathbf{W}_{k}[m,n]) \leq 1, \forall k,m,n$. The corresponding proof is similar to the one presented in Appendix~B.
\subsection{Penalized Successive Convex Approximation}
 In order to solve (\ref{eq3}) using successive convex approximation, we require a feasible initial point that satisfies QoS constraint $\mbox{C1}$. Since it is not easy to find such initial feasible points, we propose a corresponding algorithm which is based on penalizing optimization problem (\ref{eq3}) when the QoS is violated. The basic idea is to relax the considered problem by adding slack variables to  constraint $\mbox{C1}$ and penalizing the sum of the violations of the constraints. Thereby, using this technique, optimization problem (\ref{eq3}) can be rewritten in equivalent form as follows:
\begin{IEEEeqnarray}{lll}\label{eq4}&  \underset { {{\mathbf {W}}}, \mathbf {z}, \boldsymbol{\tau}}{ \mathop {\mathrm {minimize}}\nolimits }~ -[F({\mathbf{z}})-\bar{V}({\mathbf{z}})]+\beta^{(j)}\sum_{k=1}^{K}\tau_{k} \\& \ \mbox {s.t.}~\mbox {C1:}\ \nonumber F_{k}({\mathbf{z}}_{k})-\bar{V}_{k}({\mathbf{z}}_{k})+\tau_{k} \geq B_{k}, \  \forall k, \\& \qquad \nonumber \mbox {C2-C4, C6, C7},
\end{IEEEeqnarray} 
where $\beta^{(j)}$ is the penalizing weight in iteration $j$, $\tau_{k}, \forall k,$ are slack variables, and $\boldsymbol{\tau}$ is the collection of slack variables $\tau_{k}, \forall k$. \textbf{Algorithm} \ref{sco2} presents an iterative algorithm for solving (\ref{eq4}). In the first iteration, by choosing a small penalty weight $\beta^{(1)}>0$, we allow the QoS constraint to be violated such that the feasible set is large. Then, in each subsequent iteration ${j}$, we use the solution from the previous iteration as initial point, increase the penalty factor $\beta^{(j)}$,  and solve the problem again. Thus, if a feasible point exists, continuing this iterative procedure eventually yields solutions where $\tau_{k}=0, \forall k,$ holds, i.e., (\ref{eq4}) becomes equivalent to (\ref{eq3}). Otherwise, if $\tau_{k}$ does not converge to zero, the original problem is not feasible.  Moreover, a maximum value for the penalty weight $\beta_{\text{max}}$ is imposed to avoid numerical issues.

\begin{algorithm}[t]
	\caption{Penalized Successive Convex Approximation}
		\begin{algorithmic}[1]
	\STATE {Initialize:} The maximum number of iterations $J_{\text{max}}$, iteration index $j=1$, initial points  $\mathbf{{W}}^{(1)}$, $\mathbf{{z}}^{({1})}$, initial penalty factor $\beta^{(1)}\gg {1}$, $\beta_{\text{max}}$, $\eta >1$.\\
	\STATE \textbf{Repeat}\\
	\STATE Solve convex problem (\ref{eq4}) for a given $\mathbf{{W}}^{({j})}$ and $\mathbf{{z}}^{({j})}$ and store the intermediate resource allocation policy \{$\mathbf{{W}}$, $\mathbf{{z}}$\}\\
	\STATE Set ${j}={j}+1$ and update $\mathbf{{W}}^{({j})}=\mathbf{{W}}$, $\mathbf{{z}}^{({j})}=\mathbf{{z}}$, and $\beta^{(j)}=\min(\eta\beta^{(j-1)},\beta_{\text{max}})$.  \\
	\STATE \textbf{Until} convergence or ${j}=J_{\text{max}}$\\
	\STATE $\mathbf{{W}}^{*}=\mathbf{{W}}^{(j)}$,  
		\end{algorithmic} 
		\label{sco2}
\end{algorithm}
\color{black}
\section{Performance Evaluation}
In this section, we provide simulation results to evaluate the effectiveness of the proposed resource allocation design for MISO OFDMA-URLLC systems. We adopt the simulation parameters given in Table I, unless specified otherwise. In our simulations,  a single cell is considered with inner and outer radius $r_{1}=50~\textrm{m}$ and $r_{2}=250~\textrm{m}$, respectively. The BS is located at the centre of the cell. The path loss is calculated as  $35.3 + 37.6 \log_{10}(d_{k})$\cite{chsecross}, where $d_{k}$ is the distance from the BS to user $k$. The small scale fading gains between the BS and the users are modelled as independent and identically Rayleigh distributed. For simplicity, all user weights are set to $\mu_{k}=1, \forall k$. All simulation results are averaged over $1000$ realizations of the path loss and multipath fading, unless
specified otherwise.
\begin{table}[t]
	 \label{tab:table}
	\centering
			\caption{System parameters used in simulations.} 
	\renewcommand{\arraystretch}{1.4}
	\scalebox{0.6}{%
		\begin{tabular}{|c||c|}
			\hline
			Parameter & Value \\ \hline \hline 
			Number and bandwidth of sub-carriers & 64 and 15 kHz \\ \hline
			Noise power density  & -174 dBm/Hz \\ \hline
			Number of bits per packet  & 160 bits \\ \hline  
			Maximum BS transmit power $P_{\text{max}}$  &  $45$~dBm \\ \hline  
			Error tolerances $\rho$ and $\delta$ for \textbf{Algorithms} \ref{polyblocka} and $\ref{intersectiona}$ &  0.01  \\ \hline    
						Penalty factors $\beta^{(1)}$, $\beta_{\text{max}}$  for \textbf{Algorithm} \ref{sco2} &  1000, 5000, 1.1  \\ \hline
							Value of $\eta$  for \textbf{Algorithm} \ref{sco2} &  1.5  \\ \hline
						Packet error probability for user $k$&  $\epsilon_{k}=10^{-6}$						\\ \hline    
	\end{tabular}}
\end{table} 
\subsection{Performance Metric}
For performance evaluation of the system, we define the sum throughput of the system for a given channel realization as follows:
\begin{IEEEeqnarray}{lll} \label{average}
	\bar{R} = \begin{cases}\frac{1}{MN}\sum_{k=1}^{K} \Psi_{k}(\mathbf{W}_{k}),  \quad  &\text{if} \,\,   \mathbf{W} \,\, \text{is feasible} \\ 0  \qquad & \text{otherwise.} \end{cases}
\end{IEEEeqnarray}
If the optimization problem is infeasible for a given channel realization, we set the corresponding sum throughput to zero. The average system sum throughput is obtained by averaging $\bar{R}$ over all considered channel realizations.   
\subsection{Performance Bound and Benchmark Schemes}
\par  We compare the performance of the proposed resource allocation algorithm design with the following benchmark and baseline schemes\footnote{We do not show simulation results for single-input single-output (SISO) OFDMA-URLLC systems to avoid overloading the figures. We refer interested readers to \cite{ghanem1} for such results.}:
\begin{itemize}
	\setlength{\itemsep}{1pt}
	\item {\textbf{Upper bound}}: To obtain an (unachievable) performance upper bound, Shannon's capacity formula is adopted in problem (\ref{optimization1}), i.e., $V(\mathbf{W})$ and $V_{k}(\mathbf{W}_{k})$ are set to zero in the objective function and constraint $\mbox {C1}$, respectively, and all other constraints are retained. The resulting optimization problem is solved optimally and sub-optimally using modified versions of \textbf{Algorithms} \ref{polyblocka} and \ref{sco2}, respectively. The corresponding sum throughput is obtained from (\ref{average}) where $\Psi_{k}(\mathbf{W}_{k})=F_{k}(\mathbf{W}_{k})$ is used. 
		\item {\textbf{Baseline scheme 1:}} For this scheme, as for the performance upper bound, the solution
	obtained for Shannon's capacity formula is used to compute the sum throughput. However, now $\Psi_{k}(\mathbf{W}_{k})=F_{k}(\mathbf{W}_{k})$ is used in (\ref{average}), and $\Psi_{k}(\mathbf{W}_{k})=F_{k}(\mathbf{W}_{k})-V_{k}(\mathbf{W}_{k}) \geq B_{k}$ is used to check the feasibility of the solution. 
			\item {\textbf{Baseline scheme 2:}}
For this scheme, we employ  maximum ratio transmission (MRT) beamforming, where $\mathbf{w}_{k}[m,n]=\sqrt{p_{k}[m,n]}\frac{\mathbf{h}_{k}[m]}{\|\mathbf{h}_{k}[m]\|}$, and optimize the powers $p_{k}[m,n]$. The resulting optimization problem is solved using successive convex approximation employing a similar approach as for deriving \textbf{Algorithm} \ref{sco2}.
\end{itemize}
  \subsection{Simulation Results}
  In this subsection, we illustrate the effectiveness of the proposed resource allocation algorithms for MISO OFDMA-URLLC systems via computer simulations.

Figs.~\ref{convergance} and \ref{convergance2} show the convergence of the proposed optimal (\textbf{Algorithm} \ref{polyblocka}) and sub-optimal (\textbf{Algorithm} \ref{sco2}) algorithms for different numbers of sub-carriers $M$ and different numbers of users $K$. We show the system sum throughput as a function of the number of iterations for a given channel realization. As can be observed from Fig.~\ref{convergance}, the proposed optimal scheme and the proposed low-complexity scheme converge to the global optimum solution after a finite number of iterations. However, the low-complexity scheme reaches the optimal point much faster than the optimal scheme. In particular, \textbf{Algorithm} \ref{polyblocka} converges to the optimal solution after approximately 2000 and 3000 iterations for $M=12$ and $M=16$, respectively, while \textbf{Algorithm} \ref{sco2} converges in less than $5$ iterations for both $M=12$ and $M=16$. For the proposed optimal scheme, the number of iterations required for convergence increases significantly as the number of sub-carriers increase since increasing the number of sub-carriers increases the dimensionality of the search space. The convergence speed of the proposed low-complexity scheme is less sensitive to the problem size and the number of users compared to that of the optimal scheme. Furthermore, Fig.~\ref{convergance} shows that the proposed low-complexity sub-optimal scheme achieves the same performance as the optimal scheme.
 \begin{figure*}[!tbp]
	\centering
		\begin{minipage}{0.49\textwidth}
				\hspace{-0.3cm}	
	\resizebox{1\linewidth}{!}{\psfragfig{}}\vspace{-4mm}
	\caption{Convergence of the proposed optimal (\textbf{Algorithm} \ref{polyblocka}) and low-complexity sub-optimal (\textbf{Algorithm} \ref{sco2}) algorithms. $P_{\text{max}}=45$~dBm, $K=2$, $N=2$, $N_{T}=4$, $D_{1}=1$, $D_{2}=2$, and $d_{1}=d_{2}=50$\textrm{m}.}
	\label{convergance}
		\end{minipage}
	        \hfill
			\begin{minipage}{0.49\textwidth}
		\centering
		\hspace{-0.3cm}
		\resizebox{1\linewidth}{!}{\psfragfig{}}\vspace{-4mm}
		\caption{Convergence of the proposed low-complexity scheme. $P_{\text{max}}=45$~dBm, $M=64$, $N=4$, $D_{1}=D_{2}=2$, $D_{k}=4,~\forall k \neq \{1,2\}$. The users are randomly distributed within the inner and the outer radius.}
		\label{convergance2}
	\end{minipage}
\end{figure*}

In Fig.~\ref{convergance}, we chose relatively small values for $M$, $N$, $N_{T}$, and $K$ since the complexity of optimal \textbf{Algorithm} \ref{polyblocka} increases rapidly with the dimensionality of the problem. In Fig.~\ref{convergance2}, we investigate the convergence behaviour of the proposed sub-optimal scheme for larger values of these parameters. As can be observed from Fig.~\ref{convergance2}, for all considered combinations of parameter values, the proposed low-complexity suboptimal scheme requires at most 5 iterations to converge.

In Figs.~\ref{chpower1} and \ref{chpower2}, we show the average sum throughput versus the maximum transmit power at the BS. As expected, the average system sum throughput improves with increasing maximum transmit power $P_{\text{max}}$ because the SINR of the users can be enhanced by allocating more power. In particular, in Fig.~\ref{chpower1}, the proposed low-complexity scheme attains virtually the same performance as the proposed optimal scheme for all considered transmit powers. In Fig.~\ref{chpower1}, we also show results for the two baseline schemes. Baseline schemes 1 and 2 achieve lower throughputs compared to the proposed schemes. For baseline scheme 2, this performance loss is caused by the sub-optimality of the fixed beamformer. This causes the average system sum throughput to quickly saturate for transmit powers exceeding $25$~dBm. For baseline scheme 1, the resource allocation policies $\mathbf{W}$ obtained based on Shannon's capacity formula often violate constraint C1 in (\ref{optimization1}), especially for small $P_{\text{max}}$, leading to a non-feasible solution. Therefore, Shannon's capacity formula should not be used for the design of MISO OFDMA-URLLC systems, since the QoS requirements of the users cannot be guaranteed. For high $P_{\text{max}}$, for the proposed schemes, all non-zero $\Tr(\mathbf{W}_{k}[m,n])$ assume large values. Hence, the corresponding $\gamma_{k}[m,n]$ in (\ref{BT}) are large and $V_{k}(\mathbf{w}_{k})$ becomes negligible compared to $F_{k}(\mathbf{w}_{k})$. Therefore, in this case, baseline scheme 1, which  assumes $V_{k}(\mathbf{w}_{k})$ is zero, yields a similar performance as the proposed schemes.
 \begin{figure*}[!tbp]
	\centering
	\hspace{-0.2cm}	
		\begin{minipage}{0.49\textwidth}
	\resizebox{1\linewidth}{!}{\psfragfig{}}\vspace{-5mm}
	\caption{Average system sum throughput versus BS transmit power, $P_{\text{max}}$, for different resource allocation schemes. $M=16$, $N_{T}=2$, $N=2$, $K=2$, $d_{k}=50~\textrm{m}, \forall k \in \{1,2\}$, and $D_{1}=1,D_{2}=2$.}
	\label{chpower1}
		\end{minipage}
		        \hfill
		\begin{minipage}{0.49\textwidth}
					\centering
			\hspace{-0.5cm}
	\resizebox{1\linewidth}{!}{\psfragfig{}}\vspace{-5mm}
	\caption{Average system sum throughput versus BS transmit power, $P_{\text{max}}$, for different resource allocation schemes. $M=64$, $N=4$, and $K=6$. The users are randomly distributed within the inner and the outer radius.}
	\label{chpower2}
		 \end{minipage}
\end{figure*}

In Fig.~\ref{chpower2}, we show the average system sum throughput for different numbers of antennas and different delay requirements. For delay scenario $S_{0}$, none of the users has delay restrictions, i.e., $D_{k}=N=4, \forall k$. In contrast, for  delay scenario $S_{1}$, two users have strict delay constraints while the remaining users do not, i.e., $D_{1}=D_{2}=2, D_{k}=N=4, \forall k\neq \{1,2\}$. As can be observed from Fig.~\ref{chpower2}, adding more antennas at the BS considerably increases the average system sum throughput, as additional antennas offer additional degrees of freedom for resource allocation and enable more efficient and more precise beamforming. In particular, for delay scenario $S_{1}$ and transmit power $P_{\text{max}}=40$~dBm, increasing the number of antennas at the BS from $N_{T}=4$ to $N_{T}=8$ improves the average system sum throughput by 96.77~\%. Fig.~\ref{chpower2} also reveals the impact of the delay requirements on the average system sum throughput. As can be observed, the stricter delay requirements for $S_{1}$ reduce the average system sum throughput compared to $S_{0}$ because the BS is forced to allocate more power to the two delay sensitive users even if their channel conditions are poor to ensure their delay requirements are met. For example, for $P_{\text{max}}=40$~dBm and $N_{T}=8$, the strict delay requirements of $S_{1}$ decreases the upper bound and the average system throughput of the proposed scheme compared to $S_{0}$ by 6.2 and 10 bits/s/Hz, respectively.   

In Fig.~\ref{ant}, we investigate the average system sum throughput versus the number of antennas at the BS, $N_{T}$, for delay scenarios $S_{0}$ and $S_{1}$ considered in Fig.~\ref{chpower2}, and different numbers of URLLC users. As can be observed from Fig.~\ref{ant}, the average system sum throughput significantly improves as the number of antennas at the BS increases. This is due to the fact that more antennas offer additional degrees of freedom for resource allocation which leads to higher received SINRs at the users. Furthermore, the proposed scheme approaches the performance upper bound as the number of BS antennas increases since the value of $V_{k}(\mathbf{w}_{k})$ in (\ref{BT}) becomes small compared to that of $F_{k}(\mathbf{w}_{k})$ for large SINRs. Hence, the impact of finite blocklength coding on the average system sum throughput can be compensated by using large numbers of antennas at the BS. Moreover, as expected, changing the delay requirements from $S_{0}$ to $S_{1}$ reduces the throughput for all considered schemes. To compensate for this effect, the BS can increase the number of antennas in order to be able to serve the users with stricter delay requirement in a more efficient manner. Fig.~\ref{ant} also elucidates the impact of the numbers of users on the average system sum throughput. As can be seen, since the proposed scheme can exploit multi-user diversity, increasing the number of users from $K=4$ to $K=6$ increases the throughput. In contrast, baseline scheme 2 cannot support $K=6$ users for delay scenario $S_{1}$ because this scheme does not exploit all available degrees of freedom for resource allocation, and hence, the two users with strict delay requirements may lead to infeasible solutions, which has a negative impact on the average system sum throughput.       
 \begin{figure*}[!tbp]
	\centering
	\hspace{-0.4cm}	
	  	\begin{minipage}{0.49\textwidth}
	\centering
		\resizebox{1\linewidth}{!}{\psfragfig{}}\vspace{-4.5mm}
	\caption{Average system sum throughput [bits/s/Hz] vs. number of antennas at the BS. $P_{\text{max}}=45$~dBm, $M=64$, and $N=4$. The users are randomly distributed within the inner and the outer radius.}
	\label{ant}
	\end{minipage}
		        \hfill
			\begin{minipage}{0.49\textwidth}
							\centering
		\hspace{-0.5cm}
		\resizebox{1\linewidth}{!}{\psfragfig{}}\vspace{-4.5mm}
		\caption{Average system sum throughput [bits/s/Hz] vs. delay in time slots. $K=6$, $P_{\text{max}}=$~45~dBm, $M=64$,  $N=5$, and $N_{T}=6$. The users are randomly distributed within the inner and the outer radius.}
		\label{delay2}
	\end{minipage}	
 \end{figure*}

In Fig.~\ref{delay2}, we investigate the effect of different delay requirements on the average system sum throughput. We consider the following delay scenarios: $\tilde{S_{0}}=\{D_{k}=N=5, \forall k\}$ (i.e., no delay sensitive users), $\tilde{{S}_{1}}=\{D_{1}=D_{2}=D,D_{k}=N=5, \forall k\neq \{1,2\}\}$, $\tilde{{S}_{2}}=\{D_{k}=D, \forall k \in\{1,2,3,4\},D_{5}=D_{6}=N=5\}$, $\tilde{{S}_{3}}=\{D_{k}=D, \forall k \in \{1,2,3,4,5\}, D_{6}=N=5\}$. In Fig.~\ref{delay2}, we show the average system sum throughout versus delay parameter $D$. As can be observed, the average system sum throughout increases with $D$, which is due to the
fact that a large $D$ increases the feasible set of problem (\ref{optimization1}). Furthermore, the average system sum throughput decreases as the number of delay sensitive users requiring a delay of $D < N=5$ increases, since having to serve more delay sensitive users reduces the flexibility in resource allocation. Moreover, the performance of baseline scheme 2 decreases significantly if delay sensitive users are present. In particular, for baseline scheme 2, changing the delay requirements from $\tilde{S_{0}}$ to $\tilde{S_{1}}$ significantly decreases the average system sum throughput, as fixed MRT beamforming is not able to adequately support delay sensitive users.

In Fig.~\ref{users}, we show the average system sum throughput
versus the number of users for delay scenarios $S_{0}$ and $S_{1}$ considered in Fig.~\ref{chpower2}. The average system sum throughput for the proposed low-complexity scheme is close to the upper bound for small numbers of users for both considered delay scenarios. This is due to the fact that if there are only few users, they can be assigned a sufficiently large number of resource blocks to make the impact of finite blocklength coding negligible. As the number of users increases, the average system sum throughput increases due to multi-user diversity. However, at the same time, the impact of finite blocklength coding becomes more pronounced, and hence, the gap between the proposed scheme and the upper bound widens. Thus, there exists a trade-off between the performance degradation caused by short blocklengths and the performance gain induced by multi-user diversity. On the other hand, baseline scheme 2 cannot support more than $K=6$ users for delay scenario ${{{S}_{1}}}$ because this scheme does not exploit all available degrees of freedom for resource allocation.
\begin{figure}
	\centering
	\resizebox{0.49\linewidth}{!}{\psfragfig{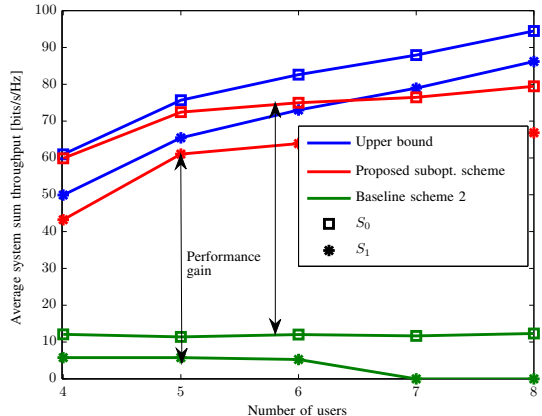}}
	\caption{Average system sum throughput versus number of users. $P_{\text{max}}=$~45~dBm, $M=64$,  $N=4$, and $N_{T}=12$. The users are randomly distributed within the inner and the outer radius.}
	\label{users}
\end{figure}

\section{Conclusion}
In this paper, we investigated the optimal resource allocation algorithm design for broadband MISO OFDMA-URLLC systems. The resource allocation algorithm design was formulated as a non-convex optimization problem for maximization of the weighted system sum throughput subject to QoS constraints for the URLLC users. The global optimum solution was obtained exploiting monotonic optimization theory. Moreover, to strike a balance between complexity and performance, we proposed a low-complexity sub-optimal algorithm to solve the optimization problem using successive convex approximation. Our simulation results revealed that the proposed sub-optimal algorithm achieves a close-to-optimal performance with low computational complexity. Furthermore, deploying multiple antennas at the BS was shown to be an effective approach to improve the reliability and to reduce the latency of URLLC systems. Moreover, our results revealed that stringent delay requirements have a negative impact on the throughput of MISO OFDMA-URLLC systems. Our results also showed that resource allocation based on Shannon's capacity formula, as is typically done in MISO OFDMA systems, may lead to infeasible solutions if URLLC is desired. Finally, the proposed optimal and sub-optimal algorithms were shown to significantly outperform two heuristic baseline schemes emphasizing the importance of optimal resource allocation in MISO OFDMA-URLLC systems.
\appendices
\section{}
  In the following, we show that the objective function of (\ref{opequivalnt}) is a difference of two monotonic and concave functions. To this end, we rewrite the objective function as follows: \vspace*{-2mm}
  \begin{IEEEeqnarray}{lll}{U}({\mathbf{z}})={F}({\mathbf{z}})-{V}({\mathbf{z}}).
   \end{IEEEeqnarray}
  ${U}({\mathbf{z}})$ is the difference of two concave functions if both ${F}({\mathbf{z}})$ and ${V}({\mathbf{z}})$ are monotonic and concave. Function ${F}(\mathbf{{z}})$ is a sum of logarithmic functions, and hence, it is a monotonic and concave function\cite{Boyed}. Furthermore, to prove that ${V}(\mathbf{{z}})$ is monotonic and concave, we rewrite it as follows: \vspace*{-2mm}
  \begin{IEEEeqnarray}{lll}
  	{V}(\mathbf{{z}})=\sum_{k=1}^{K}\mu_{k}Q^{-1}(\epsilon_{k})\sqrt{\sum_{m=1}^{M}\sum_{n=1}^{N}{V}_{k}[m,n]} ,
  \end{IEEEeqnarray}  
  where \vspace*{-2mm}
  \begin{IEEEeqnarray}{lll}
  	{V}_{k}[m,n]= a^{2}\left(1-\left({1+z_{k}[m,n]}\right)^{-2}\right).
  \end{IEEEeqnarray} 
  Note that ${V}(\mathbf{{z}})$ is always positive, because for $\epsilon_{k} \in (0,0.5)$, $Q^{-1}(\epsilon_{k}) >0$ holds. 
  To prove the monotonicity and the concavity of ${V}(\mathbf{{z}})$, first we will show that  ${V}_{k}[m,n]$ is concave by taking the first and second derivatives with respect to $z_{k}[m,n]$  as follows:  
  \begin{eqnarray}
  \frac{\mathrm{d}{V}_{k}[m,n]}{\mathrm{d}z_{k}[m,n]}&=&\frac{ 2a^{2}}{(1+z_{k}[m,n])^3}, \\
  \frac{\mathrm{d}^2{V}_{k}[m,n]}{{\mathrm{d}(z_{k}[m,n]})^2}&=&\frac{-6 a^{2}}{(1+z_{k}[m,n])^4}.
  \end{eqnarray} 
  Function ${V}_{k}[m,n]$ is a monotonic increasing and concave function because the first derivative is positive and the second derivative is negative for any $z_{k}[m,n]>0$, respectively. Moreover, since a sum of monotonic functions is monotonic, and the sum of concave functions is also concave, $\sum_{m=1}^{M}\sum_{n=1}^{N}{V}_{k}[m,n]$ is a monotonic and  concave function. By using the composition rules of convex analysis, the square root is concave and the extended-value extension on the real line is non-decreasing \cite{Boyed}. Thus, the square root of a monotonic and concave function is monotonic and concave. Finally, a weighted sum of monotonic and concave functions is also a monotonic and concave. This concludes the proof. 
\section{}
The SDP problem in (\ref{intersection}) is jointly convex in the optimization variables and satisfies Slater's constraint qualifications. Therefore, strong duality holds and solving the dual problem is equivalent to solve the primal problem\cite{Boyed}. To formulate the dual problem, we write the Lagrangian of problem (\ref{intersection}) as follows:
\begin{IEEEeqnarray}{lll}\label{optimization32}
	\mathcal{L}=-\sum_{k=1}^{K}\sum_{m=1}^{M}\sum_{n=1}^{N}\theta_{k}[m,n][f_{k}[m,n]({\mathbf {W}})-{\lambda} z^{(j)}_{k}[m,n] g_{k}[m,n]({\mathbf {W}})]\nonumber\\+\alpha \sum_{k=1}^{K}\sum_{m=1}^{M}\sum_{n=1}^{N}\Tr({\mathbf {W}}_{k}[m,n])\nonumber \\+\sum_{k=1}^{K}\sum_{n=1}^{N}\eta_{k}[n]\Tr({\mathbf {W}}_{k}[m,n])
	-\sum_{k=1}^{K}\sum_{m=1}^{M}\sum_{n=1}^{N}\Tr({\mathbf {W}}_{k}[m,n]\mathbf{Y}_{k}[m,n])+\Lambda, 
\end{IEEEeqnarray} 
where $\Lambda $ represents the collection of all terms that are independent of ${\mathbf {W}}$. Variables $\theta_{k}[m,n]$, $\alpha$, and $\eta_{k}[n]$ are the Lagrange multipliers associated with constraints $\mbox{C6}$, $\mbox{C2}$, and $\mbox{C3}$, respectively. Matrices $\mathbf{Y}_{k}[m,n] \in \mathbb{C}^{N_{T}\times N_{T}}$ are the Lagrange multipliers for the positive semi-definite constraint $\mbox{C4}$ for matrices ${\mathbf {W}}_{k}[m,n]$. Therefore, the dual problem for the SDP problem in (\ref{intersection}) is given as follows:
\begin{IEEEeqnarray}{lll}\label{dualproblem}
	\mathop{\text{maximize}}_{{\theta_{k}[m,n],\alpha,\eta_{k}[n]\geq 0,}\atop{\mathbf{Y}_{k}[m,n]\succeq 0}} \mathop{\text{minimize}}_{{\mathbf{W}}_{k}[m,n],\boldsymbol{\zeta}} \mathcal{L} ({\mathbf{W}},\boldsymbol{\zeta},\theta_{k}[m,n],\alpha,\eta_{k}[n],{\mathbf{Y}}_{k}[m,n]).
\end{IEEEeqnarray} 
In the following, we reveal the structure of the optimal ${\mathbf{W}}$ of (\ref{intersection}) by studying the Karush Kuhn Tucker (KKT) optimality conditions. The KKT conditions for the optimal solution ${\mathbf{W}}^{*}$ are given by:
\begin{IEEEeqnarray}{lll}\label{optimaldual1}
	\mathbf{Y}_{k}^{*}[m,n] \succeq 0, \quad \theta_{k}^{*}[m,n], \alpha^{*},\eta^{*}_{k}[n] \geq 0 
\end{IEEEeqnarray} 
\vspace{-1cm}
\begin{IEEEeqnarray}{lll}\label{optimaldual2}
	\mathbf{Y}_{k}^{*}[m,n] {\mathbf {W}}_{k}^{*}[m,n] = \pmb{0}, 
\end{IEEEeqnarray} 
\vspace{-1cm}
\begin{IEEEeqnarray}{lll}\label{gradient1}
	\nabla_{{\mathbf{W}}_{k}^{*}[m,n]}\mathcal{L}=\pmb{0},
\end{IEEEeqnarray} 
where $\mathbf{Y}_{k}^{*}[m,n]$, $\theta_{k}^{*}[m,n]$, $\alpha^{*}$, and $\eta^{*}_{k}[n]$ are the optimal Lagrange multipliers for dual problem (\ref{dualproblem}), and $\nabla_{{\mathbf{W}}_{k}^{*}[m,n]}\mathcal{L}$ denotes the gradient with respect to matrices ${{\mathbf{W}}_{k}^{*}}[m,n]$. The KKT condition in (\ref{gradient1}) can be rewritten as follows:
\begin{IEEEeqnarray}{lll}\label{kktcondition}
	-\theta_{k}^{*}[m,n]\mathbf{H}_{k}[m]+\alpha \mathbf{I}_{N_{T}}+\eta_{k}[n] \mathbf{I}_{N_{T}}-\mathbf{Y}_{k}^{*}[m,n]=\pmb{0}.
\end{IEEEeqnarray} 
By rearranging the terms in (\ref{kktcondition}), we obtain: 
\begin{IEEEeqnarray}{lll}\label{kktcondition2}
	(\alpha+\eta_{k}[n]) \mathbf{I}_{N_{T}}=\theta_{k}^{*}[m,n]\mathbf{H}_{k}[m]+\mathbf{Y}_{k}^{*}[m,n].
\end{IEEEeqnarray} 
Multiplying both sides of (\ref{kktcondition2}) with $\mathbf {W}_{k}^{*}[m,n]$ and exploiting (\ref{optimaldual2}), we get:
\begin{IEEEeqnarray}{lll}\label{kktcondition3}
	(\alpha+\eta_{k}[n]) \mathbf {W}_{k}^{*}[m,n]=\theta_{k}^{*}[m,n]\mathbf{H}_{k}[m,n]\mathbf {W}_{k}^{*}[m,n].
\end{IEEEeqnarray} 
Now, we consider two cases for the value of $\alpha+\eta_{k}[n]$, namely $\alpha+\eta_{k}[n]=0$ and $\alpha+\eta_{k}[n]>0$. For the first case, since both $\alpha$ and $\eta_{k}[n]$ are non-negative $\alpha+\eta_{k}[n]=0$ implies that $\eta_{k}[n]=0$, and as a result, constraint $\mbox{C3}$ holds with equality. This means that $\mathbf {W}_{k}^{*}[m,n]=\mathbf{0}$ and hence $\Rank(\mathbf {W}_{k}^{*}[m,n])$ is zero. For the second case, when $\alpha+\eta_{k}[n]>0$ holds, using basic rank inequalities for matrices, we obtain the following relations:
\begin{IEEEeqnarray}{lll}\label{rankk}
	\hspace{-2cm}
	\Rank((\alpha+\eta_{k}[n]) \mathbf {W}_{k}^{*}[m,n])=\Rank(\mathbf {W}_{k}^{*}[m,n])\nonumber\\= \Rank(\theta_{k}^{*}[m,n]\mathbf{H}_{k}[m]\mathbf {W}_{k}^{*}[m,n]) \leq \Rank(\theta_{k}^{*}[m,n]\mathbf{H}_{k}[m]) \leq 1.
\end{IEEEeqnarray} 

This implies that the beamforming matrix is either rank one or $\mathbf {W}_{k}^{*}[m,n]=\pmb{0}$, i.e., no transmission to user $k$ on subcarrier $m$ on time slot $n$. This completes the proof of Theorem 1.\color{black}
\bibliography{ref}  
\bibliographystyle{IEEEtran}
\end{document}